\documentclass[onecolumn]{article}

\usepackage{epsfig}
\usepackage{etoolbox,bm}
\usepackage{graphicx, mwe}
\usepackage{amssymb,amsmath}
\usepackage{bm,amssymb,amsfonts,epsfig,color}
\usepackage{cite,epstopdf}
\usepackage{enumerate,url,algorithm,algpseudocode}
\usepackage[algo2e]{algorithm2e}
\usepackage{multirow}
\usepackage{mathtools}
\usepackage{dblfloatfix}
\usepackage{tabularx}
\usepackage{xcolor}
\usepackage{tikz}
\usepackage{pgf-pie, booktabs}

\usepackage{amsthm}

\newcounter{define}

\newcounter{propose}

\newtheorem{proposition}[propose]{Proposition}

\newtheorem{definition}[define]{Definition}

\newcommand{\tensor}[1]{{\bm {\mathcal{#1}}}}
\newcommand{\aggtensor}[1]{{\check{\bm {\mathcal{#1}}}}}

\begin{document}
	\setlength{\abovedisplayskip}{3pt}
	\setlength{\belowdisplayskip}{3pt}
	
	\title{ GRATE: Granular Recovery of Aggregated Tensor Data by Example}

\author{Ahmed S. Zamzam 
	\and Bo Yang 
	\and Nicholas D. Sidiropoulos} 

%

\maketitle

\begin{abstract}
In this paper, we address the challenge of recovering an accurate breakdown of aggregated tensor data using disaggregation examples. {This problem is motivated by several applications.} For example, given the breakdown of energy consumption at some homes, how can we disaggregate the total energy consumed during the same period at other homes? 
In order to address this challenge, we propose GRATE, a principled method that turns the ill-posed task at hand into a constrained tensor factorization problem. Then, this optimization problem is tackled using an alternating least-squares algorithm. GRATE has the ability to handle exact aggregated data as well as inexact aggregation where some unobserved quantities contribute to the aggregated data. Special emphasis is given to the energy disaggregation problem where the goal is to provide energy breakdown for consumers from their monthly aggregated consumption. Experiments on two real datasets show the efficacy of GRATE in recovering more accurate disaggregation than state-of-the-art energy disaggregation methods.
\end{abstract}
	
	\section{Introduction}
\label{sec:introduction}

Data aggregation occurs in many applications in order to reduce data size, protect fine granularity data, or summarize datasets.  One example is sales data, which are usually aggregated {\em temporally} over quarters~\cite{di1990estimation}. Another example is electric power utility data, where due to the sensitivity of supply and demand data, utility companies often release only {\em geographically} aggregated data as a compromise to the research community~\cite{anderson2018disaggregation}. Aggregated data also appear in applications where high resolution data is unavailable due to lack of precision measurement tools.

Irrespective of the reason why data is only available in aggregated form, recovering data at a finer granularity is valuable for many data mining tasks. Better understanding of individual behaviors can be gained with detailed data~\cite{lenzen2011aggregation}. In addition, using aggregated data in order to describe individual behavior may result in misleading conclusions~\cite{clark1976effects,garrett2003aggregated}. Hence, many research efforts have focused on recovering an accurate breakdown of aggregated data~\cite{chow1971best, rossi1982note, faloutsos1997recovering, liu2017h, Alistair2018homerun}. In order to reconstruct data from aggregated reports, underlying properties of the data, e.g.,  smoothness~\cite{chow1971best,almutairi2018scalable}, periodicity~\cite{faloutsos1997recovering}, or both~\cite{liu2017h, Alistair2018homerun}, are often exploited.

The monthly electricity bill represents a simple aggregation example where the total energy consumption is the sum of the consumption of many devices. Providing customers with a breakdown of their energy consumption is expected to yield $ 12 - 15\% $ savings~\cite{armel2013disaggregation}. At first glance, it may seem impossible to estimate the device-level energy breakdown accurately given only the aggregate monthly consumption. However, the rapid proliferation of smart meters installed in homes {sheds} light on how much each device contributes to the aggregated energy consumption in {\em different types of homes}. Unfortunately, most homes are not equipped with device-level smart meters, making their monthly aggregated consumption the only data available. Utilizing similarity between homes in terms of energy consumption patterns, several approaches have been proposed to estimate the energy breakdown from aggregated data using examples of detailed data from different homes~\cite{batra2016gemello,batra2017matrix,batra2018transferring}. However, these methods do not exploit the distinct structure of the data and the similarity in consumption patterns over years, due to seasonal periodicities.  

We tackle this disaggregation problem using tensor methods.
Tensors are a natural generalization of matrices, where each tensor element is indexed by three or more indices -- as opposed to just two (row index and column index) for matrices. {Due to their natural suitability for modeling multi-relational data,} tensor factorization methods have been extensively used for data mining and analysis tasks~\cite{kolda2005higher,bader2007efficient,karatzoglou2010multiverse, almutairi2019prema}; See~\cite{sidiropoulos2017tensor} for an overview of tensor algebra and applications. The focus of this paper is on disaggregating tensor data. For example, the {aforementioned} energy disaggregation problem for multiple homes can be formalized as a tensor disaggregation problem where each element represents the energy consumption of an appliance at a certain home during a specific month~\cite{batra2018transferring}. Disaggregating tensor data was also considered in~\cite{almutairi2019prema} where {\em multiple} multi-dimensional aggregated views were used to recover tensor data. In this paper, the special algebraic structure of the aggregation mechanism along one dimension is exploited in order to recover more accurate energy breakdowns. 

\vspace{4pt}
\noindent{{\textsc{Informal Problem Definition}}
\vspace{2pt}

The considered task is informally defined as follows.
\begin{itemize}
	\item {\em Given:} aggregated tensor data along one dimension and some detailed examples.\vspace{2pt}
	\item {\em Obtain:} finer-granularity (appliance-level) data.
\end{itemize}

We introduce GRATE, a systematic and efficient approach to address the problem of disaggregating tensor data given some detailed examples. First, the approach constructs a compound tensor by appending the aggregated data to the original tensor along the aggregation mode. Then, the distinct structure of this tensor is translated into constraints on the latent factors of this compound tensor. Although the constraints are nonlinear, they are linear in each latent factor separately which makes the factorization problem amenable to alternating linearly constrained linear least squares solvers. Enforcing these constraints while factoring the data exploits the nature of the data and leads to a more accurate disaggregation. Special emphasis is given to estimating the breakdown of monthly aggregated electricity consumption where the data are represented as a four-way tensor (appliance $ \times $ home $ \times $ month $ \times $ year).

The effectiveness of the GRATE approach is demonstrated using a dataset of breakdowns of energy consumption of many homes in Austin, TX. GRATE exhibits superior estimation accuracy in both exact and inexact (dis)aggregation cases. To summarize, the contributions of the paper are: {\bf c1)} introduction of a constrained tensor factorization problem formulation that respects the particular structure of aggregated tensor data, {\bf c2)}  in order to utilize the correlation between seasons of different years, the energy disaggregation data is modeled using a four-way tensor which enhances the performance over methods using three-way representations, {\bf c3)} careful experimental evaluation using real datasets, which shows that the proposed approach clearly outperforms prior state-of-the-art methods.



	\section{Preliminaries and Problem Formulation}
\label{sec:prob_form}
\subsection{Problem Definition}
Without loss of generality, let us focus on the case of a four-mode tensor. Consider an aggregated three mode tensor $ \tensor{X}_{1} $, which aggregates the tensor elements along the first mode, i.e., the element $ \tensor{X}_{1}(i_2, i_3, i_4) $ is given by
\begin{align}
\tensor{X}_1(i_2, i_3, i_4) = \sum_{i_1 = 1}^{I_1} \tensor{X}(i_1, i_2, i_3, i_4).
\end{align}
Appending this tensor to the original tensor $ \tensor{X} $ results in a four mode tensor $ \aggtensor{X} $ of size $ (I_1 + 1) \times I_2 \times I_3 \times I_4 $. Additional aggregation along different modes can be accounted for as well. However, for the sake of simplicity, this section focuses on a single aggregation along the first mode of the tensor. 

This aggregation structure is present in many applications. 
As an example, in energy consumption data, we have access to the aggregated consumption of all users at any point of time, while the detailed data can be modeled as a tensor. An entry of this tensor represents the energy consumption of an appliance at a home during a time interval~\cite{batra2016gemello}.
Detailed description of the energy disaggregation problem will be presented in Section~\ref{sec:energy}. 

Similar scenarios appear in healthcare~\cite{liu2017h} and retail analytics~\cite{di1990estimation}. Oftentimes, only a subset of the entries of $ \aggtensor{X} $ is observable due to application-specific constraints such as privacy concerns, {or unavailability of fine-grained measuring devices}, or experimental errors affecting the data collection process. In the considered aggregated tensors, most of the misses are expected to be in the form of {fibers} along the aggregation mode of $ \tensor{X} $ as shown in Fig.~\ref{fig:agg_tensor}, i.e., for some $ (i_2, i_3, i_4) $ tuples, the element $ (I_1+1, i_2, i_3, i_4) $ is given but the mode-$ 1 $ fiber of $ \tensor{X} $ indexed by $ (i_2, i_3, i_4) $ is missing. Therefore, $ \aggtensor{X} $ is expressed as $ \aggtensor{X} = \aggtensor{X}^{A} + \aggtensor{X}^M $, where $ \aggtensor{X}^A $ holds the available tensor entries of $ \aggtensor{X} $ and elsewhere is zero and $ \aggtensor{X}^M $ contains the missing entries and is zero elsewhere.

{This work considers two aggregation scenarios, for which it proposes different factorization approaches.}
\begin{enumerate}
	\item {\em Exact aggregation}: in this case, the elements on the mode-$ 1 $ fiber indexed by $ (i_2, i_3, i_4)$ are summed up exactly in the element $ \aggtensor{X}(I_1+1, i_2, i_3, i_4) $.
	\item {\em Inexact aggregation}: in this case, some unobserved positive quantity is added to the aggregation. This happens in the energy disaggregation problem where some unmodeled devices (e.g., light bulbs, small kitchen appliances, portable heaters) contribute to the reported aggregate amount. This can be written as
	\begin{align}\label{eq:inexact_agg}
	\aggtensor{X}(I_1+1, i_2, i_3, i_4) \geq \sum_{i_1=1}^{I_1} \tensor{X}(i_2, i_3, i_4)
	\end{align}
	for any $ i_2, i_3, $ and $ i_4 $.
\end{enumerate}

\vspace{6pt}
\fbox{%
	\parbox{\textwidth- 40pt}{%
		{\bf Problem Statement.} Given a partially observed tensor $ \aggtensor{X} $, our goal is to derive a factorization model that utilizes the special structure of $\aggtensor{X}$ 
		in order to estimate missing data $ \aggtensor{X}^M $ and break down the sums along mode-$ 1 $. 
	}%
}
\vspace{12pt}

Solving the above problem offers far-reaching benefits to our economy and environment. Several studies~\cite{OPower,armel2013disaggregation} have shown that providing itemized electricity bills to the customers fosters better understanding of the energy consumption. Besides, discrepancy in disaggregated consumption data helps in identifying defective devices, which eventually leads to more savings. In addition, utility companies can use disaggregated data to enable better energy recommendations~\cite{beckel2013automatic}. Finally, such detailed data enable more efficient demand response programs and accurate energy consumption forecasting.


Nonetheless, it is challenging to estimate the breakdown of energy consumption given only the aggregated consumption data. 
The availability of detailed data from {a few} homes that are equipped with smart meters provides some hope for solving this problem, as these detailed data can serve as examples on how to disaggregate. 
Yet, these examples are not perfect, due to the presence of unmodeled devices even in the homes equipped with smart meters, and missing data due to measuring errors (e.g., sensor fails to report). This fact poses additional challenges to the energy disaggregation problem.

We tackle these challenges with specifically designed tensor factorization methods. Tensor factorization is appealing as it directly captures the multi-relation structure in the energy consumption data. More importantly, tensor factorization enables easy incorporation of physical constraints, which in turn ensures validity of the results -- this point will become clear in Section~\ref{sec:approach}. Before presenting the proposed method, we briefly review some necessary background on tensor factorization.
\subsection{Tensor Decomposition Preliminaries}

A tensor is a multi-dimensional array indexed by three or more indices. An $ N $-way tensor $ \tensor{X} \in \mathbb{R}^{I_1 \times \cdots \times I_N}$ of rank $ F $ can be represented as
\begin{align}\label{eq:tensor-cpd}
\tensor{X}(i_1, i_2, \cdots, i_N) = \sum_{f=1}^{F}\prod_{n=1}^N{\bf A}_n(i_n, f)
\end{align}
where $ {\bf A}_n \in \mathbb{R}^{I_n \times F} $ is the $ n $-th latent matrix of the canonical polyadic decomposition (CPD) of $ \tensor{X} $. The low rank tensor decomposition decomposes $ \tensor{X} $ into a sum of outer products of columns of $ {\bf A}_1, \ldots, {\bf A}_N $, where each outer product is a rank-one tensor.

The $ n $-th mode unfolding of $ \tensor{X} $ is denoted by $ {\bf X}_{(n)} $, and is obtained by stacking the vectorization of the mode-$ n $ slabs to form a matrix. The $ n $-th mode unfolding can be equivalently written as
\begin{align}
{\bf X}_{(n)} = ({\bf A}_1 \!\odot\! \cdots \!\odot\! {\bf A}_{n-1} \!\odot\! {\bf A}_{n+1} \!\odot\! \cdots \!\odot\! {\bf A}_{N})\ {\bf A}_n^T,
\end{align}
where $\odot$ stands for the Khatri-Rao (column-wise Kronecker) matrix product. We also write $\tensor{X} = \big[[{\bf A}_1, {\bf A}_2, \ldots, {\bf A}_N ]\big]$ to denote that $\tensor{X}$ is the synthesized tensor as in~\eqref{eq:tensor-cpd}.

For a matrix $ {\bf A} $, the notion of Kruskal rank or k-rank will be used in the manuscript. The Kruskal rank of a matrix is defined as follows.
\begin{definition}[Kruskal Rank]
	The Kruskal rank $ k_{\bf A} $ of an $ I \times F $ matrix $ {\bf A} $ is the largest integer $ k $ such that any $ k $ columns of $ {\bf A} $ are linearly independent.
\end{definition}



	\section{Proposed Approach}
\label{sec:approach}
In this section, we present our solution of the problem stated above. 
First, we derive the appropriate constraints brought by data aggregation. Then, the formal problem formulation and algorithm are presented in Section~\ref{sec:alg}.
\subsection{Constraints for disaggregation}
The CPD model of the compound tensor $ \aggtensor{X} $ can be expressed as
\begin{align}
\aggtensor{X} = \sum_{r=1}^{R} \check{{\bf A}}_1(:, r) \circ {\bf A}_2(:, r) \circ  {\bf A}_3(:, r) \circ {\bf A}_4(:, r)  
\end{align}
where the matrix $ \check{\bf A}_1 \in \mathbb{R}^{(I_1+1)\times R} $ represents the first mode latent factor of $ \aggtensor{X} $. 
{Now, let us define two vectors of ${\bf e}_s, {\bf e}_b \in \mathbb{R}^{I_1 +1}$ as follows
\begin{align}
{\bf e}_s &= [ \boldsymbol{0}_{I_1}, \quad  1]^T,\\
{\bf e}_b &= [ \boldsymbol{1}_{I_1}, \quad 0]^T.
\end{align}}
Also, we define a $ 4 $-way tensor $ \tensor{G}\in\mathbb{R}^{R \times I_2 \times I_3 \times I_4} $ as follows
\begin{align}\label{key}
\tensor{G}(r, i_2, i_3, i_4) = {\bf A}_2 (i_2, r) {\bf A}_3 (i_3, r) {\bf A}_4 (i_4, r).
\end{align}
Additionally, let the matrix $ {\bf G}_{(1)} $ be the mode-$ 1 $ matricization of $ \tensor{G} $. Next, we derive the latent domain constraints required for factoring the exact and inexact aggregated tensor. 

\begin{proposition}
	\label{prop:exact_equi}
	For a tensor $ \aggtensor{X} \in \mathbb{R}^{(I_1 +1) \times I_2 \times I_3 \times I_4} $ such that
	\begin{align}\label{eq:aggr}
	\aggtensor{X}(I_1+1, i_2, i_3, i_4) = \sum_{i_1 = 1}^{I_1} \aggtensor{X}(i_1, i_2, i_3, i_4) 
	\end{align}
	for all $ i_2, i_3, $ and $ i_4 $, the low-rank CPD factors of $ \aggtensor{X} $ satisfy
	{  \begin{align}
	{\bf G}_{(1)}\ \check{{\bf A}}_1^T ({\bf e}_s - {\bf e}_b) = {\bf 0}.
	\end{align}
	In addition, if the matrix $ {\bf G}_{(1)} $ is full rank, then $ \check{{\bf A}}_1^T ({\bf e}_s - {\bf e}_b) = {\bf 0}$ is equivalent to \eqref{eq:aggr}.}
\end{proposition}
Note that \eqref{eq:aggr} means that the last slab of tensor $\aggtensor{X}$ is the aggregation of all the slabs that precede it.
\begin{proof}
	Fixing the indices $ i_2 $ through $ i_4 $, we have
	\begin{align*}
	\aggtensor{X}(I_1\!+\!1, i_2, i_3, i_4) = \sum_{i_1 = 1}^{I_1} \aggtensor{X}(i_1, i_2, i_3, i_4) \Leftrightarrow 	
	\end{align*}
	\begin{align*}
	\begin{aligned}
	&\sum_{r =1}^{R}\! \check{\bf A}_1(I_1\!+\!1, r) {\bf A}_2(i_2, r) {\bf A}_3(i_3, r){\bf A}_4(i_4, r)\! =\\
	&\qquad \qquad \qquad \qquad \sum_{i_1 = 1}^{I_1}\!  \sum_{r =1}^{R}\! \check{\bf A}_1(i_1,\! r) {\bf A}_2(i_2,\! r) {\bf A}_3(i_3,\! r) {\bf A}_4(i_4, r) \Leftrightarrow\\
	\end{aligned}
	\end{align*}
	\begin{align*}
	\begin{aligned}
	&\sum_{r =1}^{R} {\bf A}_2(i_2, r) {\bf A}_3(i_3, r) {\bf A}_4(i_4, r)  \check{\bf A}_1(I_1+1, r) =\\
	&\qquad \qquad \qquad \qquad \sum_{r =1}^{R}\!  {\bf A}_2(i_2, r) {\bf A}_3{i_3, r} {\bf A}_4(i_4, r) \sum_{i_1 = 1}^{I_1}\! \check{\bf A}_1(i_1, r)\Leftrightarrow
	\end{aligned}
	\end{align*}
	{ \begin{align*}
	{\bf g}_{i_2, i_3, i_4}^T \check{{\bf A}}_1^T {\bf e}_s &= {\bf g}_{i_2, i_3, i_4}^T \check{{\bf A}}_1^T {\bf e}_s\\		\notag
	{\bf g}_{i_2, i_3, i_4}^T \check{{\bf A}}_1^T ({\bf e}_s -{\bf e}_b) &= {\bf 0}.
	\end{align*} 
	where $ {\bf g}_{i_2, i_3, i_4} $ is a vector of length $ R $ and holds $ \big({\bf A}_2(i_2, r) {\bf A}_3(i_3, r) {\bf A}_4(i_4, r)\big)$ at the $ r $-th location.
	Applying this to all possible combinations of indices $ i_2, i_3, i_4 $, we obtain
	\begin{align}\label{eq:aggr_latent}
	{\bf G}_{(1)} \check{{\bf A}}_1^T ({\bf e}_s - {\bf e}_b) = {\bf 0}.
	\end{align}
	Indeed, if the matrix $ {\bf G}_{(1 )} $ is full rank, then \eqref{eq:aggr_latent} implies $\check{{\bf A}}_1^T {\bf e}_s = \check{{\bf A}}_1^T {\bf e}_b$. In addition, it is easy to see that $\check{{\bf A}}_1^T {\bf e}_s = \check{{\bf A}}_1^T {\bf e}_b$ implies \eqref{eq:aggr_latent}, which is equivalent to \eqref{eq:aggr}.	
	Therefore, $\check{{\bf A}}_1^T ({\bf e}_s - {\bf e}_b) = \boldsymbol{0}$ is equivalent to \eqref{eq:aggr} when $ {\bf G}_{(1 )} $ is full rank.}
\end{proof}

{ The above proposition asserts that when ${\bf G}_{(1)}$ is full rank, the latent constraint $\check{{\bf A}}_1^T ({\bf e}_s - {\bf e}_b) = \boldsymbol{0}$ is equivalent to the constraint \eqref{eq:aggr}. 
This is important, as we ultimately would like to enforce constraints on latent factors -- this is much easier. When the condition is not satisfied, there may exist ${\bf A}_1$ which satisfy $\check{{\bf A}}_1^T ({\bf e}_s - {\bf e}_b) \neq \boldsymbol{0}$ while \eqref{eq:aggr_latent} still holds. 
In such a case, enforcing $\check{{\bf A}}_1^T ({\bf e}_s - {\bf e}_b) = \boldsymbol{0}$ can result in sub-optimal performance as it restricts the solution space. Proposition~\ref{prop:exact_equi} shows that enforcing the latent constraint $\check{{\bf A}}_1^T ({\bf e}_s - {\bf e}_b) = \boldsymbol{0}$ is without loss of generality and optimality when ${\bf G}_{(1)}$ is full rank.}

Next, we derive a condition on the latent factors which ensures that $ {\bf G}_{(1)} $ is full rank. Notice that $ {\bf G}_{(1)} $ can also be written as the Khatri-Rao product of the latent factors of the dimensions that are not aggregated, i.e., $ {\bf G}_{(1)} = ({\bf A}_2 \odot {\bf A}_3 \odot {\bf A}_4)$.
\begin{proposition}
	The matrix $ {\bf G}_{(1)} \in \mathbb{R}^{I_2I_3I_4 \times R} $ is full column rank if 
	\begin{equation}\label{eq:fullrank-cond}
		k_{{\bf A}_2} + k_{{\bf A}_3} + k_{{\bf A}_4} \geq R + 2.
	\end{equation}
\end{proposition}
\begin{proof}
	Since $ {\bf G}_{(1)} = ({\bf A}_2 \odot {\bf A}_3 \odot {\bf A}_4)$, using the bound on the Kruskal rank of a Khatri-Rao product from~\cite{sidiropoulos2000uniqueness, ten2000k}, we obtain
	\begin{align}
		k_{{\bf G}_{(1)}} &= \min \{ \min \{ k_{{\bf A}_2} + k_{{\bf A}_3} - 1 , R \} + k_{{\bf A}_4} - 1, R \}\nonumber\\
			&= \min \{ \min \{ k_{{\bf A}_2} + k_{{\bf A}_3} + k_{{\bf A}_4} - 2 , R + k_{{\bf A}_4} - 1\}, R \}\nonumber\\
			&= \min \{ k_{{\bf A}_2} + k_{{\bf A}_3} + k_{{\bf A}_4} - 2 , R \}\nonumber
	\end{align}
	By definition, $ \text{rank}({\bf G}_{(1)}) \geq k_{{\bf G}_{(1)}} $. Therefore, if $ k_{{\bf A}_2} + k_{{\bf A}_3} + k_{{\bf A}_4} \geq R + 2 $, then $ k_{{\bf G}_{(1)}} \geq R$ and ${\bf G}_{(1)}$ is full column rank.
\end{proof}

Notice that the dimension of $ {\bf A}_n $ is $ I_n \times R $ where $ I_n $ is often larger than $ R $ for real world data. Therefore, the condition~\eqref{eq:fullrank-cond} is rather mild. Using the same line of argument as in \cite{JiangSidtBerge2001} with real (decaying) in place of complex (constant modulus) exponentials, it can be shown that $I_2 I_3 I_4 \geq R$ is sufficient almost surely -- i.e., for almost all ${\bf A}_2,{\bf A}_3,{\bf A_4}$ (except for a set of measure zero in $\mathbb{R}^{(I_2 + I_3 + I_4)R}$).
{ Hence, in the proposed algorithm for exact aggregation, the constraint $ \check{{\bf A}}_1^T ({\bf e}_s - {\bf a}_b) = \boldsymbol{0} $ is used to ensure that the factorization model satisfies the aggregation structure of $ \aggtensor{X} $. }

Next, we consider the inexact aggregation case where the reported aggregate amount includes consumption from unmodeled appliances. This type of aggregation appears in the energy disaggregation scenario which will be discussed in Section~\ref{sec:energy}.

\begin{proposition}
	For a tensor $ \aggtensor{X} \in \mathbb{R}_+^{(I_1 +1) \times I_2 \times I_3 \times I_4} $ such that
	\begin{align}\label{eq:inexact}
	\aggtensor{X}(I_1+1, i_2, i_3, i_4) \geq \sum_{i_1 = 1}^{I_1} \aggtensor{X}(i_1, i_2, i_3, i_4) 
	\end{align}
	for all $ i_2 $, $ i_3 $, and $ i_4 $, the low-rank CPD factors of $ \aggtensor{X} $ satisfy
	{\begin{align}\label{eq:geq-cond}
	{\bf G}_{(1)} \check{{\bf A}}_1^T ({\bf e}_s - {\bf e}_b) \geq {\bf 0}.
	\end{align}}
\end{proposition}
\begin{proof}
	The proposition can be proven by following the steps of the proof of Proposition~\ref{prop:exact_equi} by replacing every equality ($ = $) with an inequality ($ \geq $).
\end{proof}

{Interestingly, the condition~\eqref{eq:geq-cond} is not equivalent to $ \check{{\bf A}}_1^T ({\bf e}_s - {\bf e}_b) \geq {\bf 0} $ even if the latent factors are nonnegative. Even though enforcing $ \check{{\bf A}}_1^T ({\bf e}_s - {\bf e}_b) \geq \boldsymbol{0} $ in the case of nonnegative latent factors ensures~\eqref{eq:inexact}, it hinders the ability of the model to fit the data.} The previous proposition asserts that enforcing~\eqref{eq:geq-cond} is equivalent to~\eqref{eq:inexact}, and hence, can be enforced on the latent domain factors to enforce the specific structure of $ \aggtensor{X} $ without restricting the model.

\subsection{Proposed algorithm: GRATE}
\label{sec:alg}
Considering the exact aggregation case, a common way to estimate the missing values of the tensor is to identify its latent factors by adopting a least-squares criterion. The proposed approach constraints the latent domain factors according to the structure of the compound aggregation tensor. Therefore, assuming that the condition~\eqref{eq:fullrank-cond} is satisfied, we obtain the following optimization task
\begin{subequations}\label{eq:optimization-exact}
	\begin{align}
	&\underset{\check{\bf A}_1, {\bf A}_2, {\bf A}_3, {\bf A}_4}{\min}\quad \bigg\| \mathcal{P}_{\Omega} \Big(\! \aggtensor{X}\! -\! \big[[\check{\bf A}_1,\! {\bf A}_2,\! {\bf A}_3,\! {\bf A}_4]\big] \!\Big)\bigg\|_F^2\label{eq:cost-fun-exact}\\ 	\label{eq:constraint-exact}
	&\ \ \text{subject to}\qquad  {\check{{\bf A}}_1^T ({\bf e}_s - {\bf e}_b) = {\boldsymbol{0}}}
	\end{align}
\end{subequations}
where the cost function penalizes the fitting error of the CPD model, and the operator $ \mathcal{P}_{\Omega} $ ensures that the loss function only includes the fitting errors in the observed entries of the tensor. Note that, the optimization problem~\eqref{eq:optimization-exact} is a standard CPD fitting problem~\cite{sidiropoulos2017tensor} with an equality constraint on one of the factors. In order to solve this problem, we use the standard ALS framework while the constraint~\eqref{eq:constraint-exact} is taken into account while solving for the factor $ \check{\bf A}_1 $. This constraint can be also accounted for by projecting the feasibility space onto the null space of the linear equation~\eqref{eq:constraint-exact}. That is, if the $ j $-th row of $ \check{\bf A}_1 $ is denoted by $ \check{\bf a}_{1,j}^T $, then we write $ \check{\bf A}_1 $ in the optimization problem as
$
\check{\bf A}_1 = [\check{\bf a}_{1,1}^T\quad \ldots\quad \check{\bf a}_{1,I_1}^T\quad \sum_{i=1}^{I_1}\check{\bf a}_{1,i}^T]^T
$. Regularization terms can be also readily incorporated into problem~\eqref{eq:optimization-exact} to promote low rank~\cite{yang2016tensor} or sparse solutions~\cite{sidiropoulos2017tensor}. Further, effective and fast convergent alternating iteratively weighted least-squares solvers have been developed to deal with the resulting non-smooth regularization terms. 

Now, we focus on the inexact aggregation scenario which is often encountered in the energy disaggregation example. Given the structure of the tensor, we develop a constrained tensor factorization algorithm to infer the latent factor matrices, and therefore, estimate $ \aggtensor{X}^M $. Toward this end, consider the following optimization task
\begin{subequations}\label{eq:optimization}
	\begin{align}
&\underset{\check{\bf A}_1, {\bf A}_2, {\bf A}_3, {\bf A}_4}{\min}\quad \bigg\| \mathcal{P}_{\Omega} \Big(\! \aggtensor{X}\! -\! \big[[\check{\bf A}_1,\! {\bf A}_2,\! {\bf A}_3,\! {\bf A}_4]\big] \!\Big)\bigg\|_F^2\label{eq:cost-fun}\\ 	\label{eq:constraint}
&\ \ \text{subject to}\quad {\big({\bf A}_2 \odot {\bf A}_3 \odot {\bf A}_4 \big) \check{{\bf A}}_1^T ({\bf e}_s - {\bf e}_b) \geq {\boldsymbol{0}}}
	\end{align}
\end{subequations}
where the cost function is the least-squares fitting error of the CPD model, and the operator $ \mathcal{P}_{\Omega}(\cdot) $ accounts only for the observed elements $ \aggtensor{X} $. The inequality constraint~\eqref{eq:constraint} reflects the special structure of the aggregated tensor.

The optimization problem in~\eqref{eq:optimization} is nonconvex due to the multilinear term in the cost function, and the nonlinear constraint. Next, we develop an efficient algorithm to solve~\eqref{eq:optimization} based on alternating least squares (ALS).

Notice that the optimization problem~\eqref{eq:optimization} is a constrained linear least squares optimization for each latent factor, when all other factors are kept fixed. The GRATE approach utilizes the ALS framework which updates each latent factor separately in a cyclic fashion. Let $ \check{\bf x}_n $ and $ {\bf a}_n $ denote the vectorization of $ {\bf X}_{(n)} $ and the latent factor matrix $ {\bf A}_n $, respectively. 
{Define $ {\bf e} \in \mathbb{R}^{I_1+1} $ such that $ {\bf e} := {\bf e}_b - {\bf e}_s$.}  Also, we denote by ${\bf D}_{\bf w}$ a diagonal matrix that holds the elements of the row vector ${\bf w}$ on the diagonal.
Algorithm~\ref{alg:als} summarizes the proposed iterative approach to solve~\eqref{eq:optimization}.

\SetKw{inputA}{Input: }
\SetKw{outputA}{Output: }
\SetKw{initA}{Initialization: }
\DontPrintSemicolon
\begin{algorithm}[h]
	\renewcommand{\arraystretch}{1.2}
	\caption{GRATE algorithm for inexact disaggregation}
	{
		\inputA{ $ \aggtensor{X}^A $}\;\\
		\outputA{$ \aggtensor{X}^M $, $ \check{\bf A}_1 $, $ {\bf A}_2 $, $ {\bf A}_3 $}, and $ {\bf A}_3 $\;\\
		\initA{Intialize $ {\bf A}_2 $, $ {\bf A}_3 $, and $ {\bf A}_4 $ randomly}\;\\
		\Repeat{cost converges}{
			\begin{align}
			\check{\bf A}_1 \Leftarrow& \  \underset{\check{\bf A}_1}{\arg\min} \big\| \mathcal{P}_{\Omega} \big(\check{\bf X}_{(1)} - ({\bf A}_2 \odot {\bf A}_3 \odot {\bf A}_4)\ \check{\bf A}_1^T \big) \big\|_2^2\nonumber\\[-5pt]
			&\ \text{s. to}\quad ({\bf A}_2 \odot {\bf A}_3 \odot {\bf A}_4) \ \check{{\bf A}}_1^T\ {\bf e}\ \!\leq\ \!{\bf 0}\nonumber
			\end{align}
			\begin{align}
			{\bf A}_2 \Leftarrow& \  \underset{{\bf A}_2}{\arg\min} \big\| \mathcal{P}_{\Omega} \big(\check{\bf X}_{(2)} - (\check{\bf A}_1 \odot {\bf A}_3 \odot {\bf A}_4)\ {\bf A}_2^T \big) \big\|_2^2\nonumber\\[-5pt]
			&\ \text{s. to}\quad\  \big(({\bf A}_3 \odot {\bf A}_4) {\bf D}_{{\bf e}^T \check{\bf A}_1}\big) \ {\bf A}_2^T\ \!\leq\ \!{\bf 0}\nonumber
			\end{align}
			\begin{align}
			{\bf A}_3 \Leftarrow& \  \underset{{\bf A}_3}{\arg\min} \big\| \mathcal{P}_{\Omega} \big(\check{\bf X}_{(3)} - (\check{\bf A}_1 \odot {\bf A}_2 \odot {\bf A}_4)\ {\bf A}_3^T \big) \big\|_2^2\nonumber\\[-5pt]
			&\ \text{s. to}\quad\ \big(({\bf A}_2 \odot {\bf A}_4) {\bf D}_{{\bf e}^T \check{\bf A}_1}\big) \ {\bf A}_3^T\ \!\leq\ \!{\bf 0}\nonumber
			\end{align}
			\begin{align}
			{\bf A}_4 \Leftarrow& \ \underset{{\bf A}_4}{\arg\min} \big\| \mathcal{P}_{\Omega} \big(\check{\bf X}_{(4)} - (\check{\bf A}_1 \odot {\bf A}_2 \odot {\bf A}_3)\ {\bf A}_4^T \big) \big\|_2^2\nonumber\\[-5pt]
			&\ \text{s. to}\quad\ \big(({\bf A}_2 \odot {\bf A}_3) {\bf D}_{{\bf e}^T \check{\bf A}_1}\big) \ {\bf A}_4^T\ \!\leq\ \!{\bf 0}\nonumber
			\end{align}	
		}
	}
	\label{alg:als}
\end{algorithm}

Using the ALS approach, a constrained linear LS problem is solved per iteration which can be efficiently solved using off-the-shelf interior point method solvers. In addition, if the latent factor matrices are nonnegative, the nonnegativity constraint can be added to the subproblems without altering the nature of the subproblems. In the next section, we present the main application that we consider in this paper.	
	\section{Energy Disaggregation}
\label{sec:energy}



{\emph{Energy disaggregation}} (energy breakdown) refers to the process of inferring the energy consumption of each individual device or appliance from the aggregated energy consumption\footnote{With a slight abuse of notation, we use `device' and `appliance' interchangeably throughout the paper.}. The disaggregation process can be performed at different temporal resolutions. For example, the aggregate energy time series at high frequency (seconds or minutes) can be disaggregated into appliance-wise consumption. Due to the nature of data in this case, time series analysis tools are widely used to breakdown the aggregate energy time series~\cite{kim2011unsupervised,zhong2014interleaved,almutairi2018scalable}. In this paper, we focus on disaggregating energy consumption at lower frequency (months), 
as this is the most basic feedback given to consumers. That is, for a monthly consumption of $ 1000 $ kWh, the proposed approach may suggest that the HVAC contribution is $ 475 $ kWh while the fridge monthly consumption is $ 60 $ kWh, and so on. 
Providing energy breakdown facilitates informed decision making. When the consumers have access to the energy consumption of each device, they are likely to reduce their consumption and identify misconfigured appliances~\cite{katipamula2005methods}. 
Unfortunately, installing specific hardware at each home in order to obtain energy breakdown is infeasible. {We develop a non-intrusive load monitoring (NILM) approach, which tries to estimate the appliance-wise consumption from a single meter measuring the aggregated home energy consumption.} Utilizing disaggregated data from relatively few homes with per-appliance measurements, we aim at achieving accurate disaggregation of the consumption at homes with a single meter. The proposed algorithm builds on our constrained tensor factorization approach presented in the previous section to breakdown the aggregated monthly energy consumption.


\begin{figure}[t]
	\centering
	\includegraphics[width=0.75\textwidth]{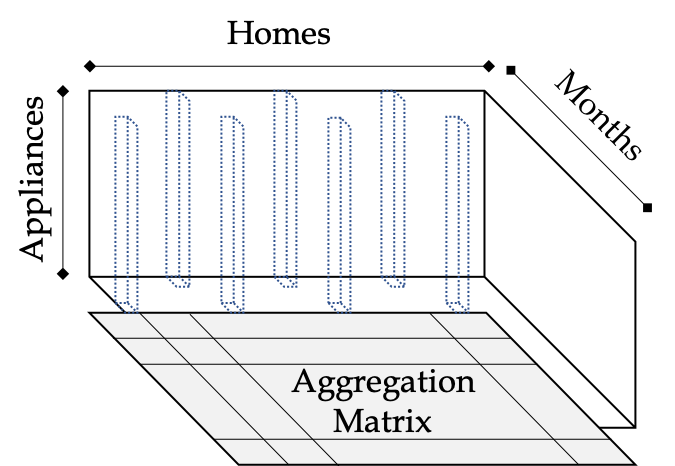}
	\caption{Example of a single year tensor. }
	\label{fig:agg_tensor}
\end{figure}

We construct a $ 4 $-way tensor that includes the monthly detailed consumption at different homes. The dimensions of the tensor represent devices, homes, months, and years. This way we can also exploit yearly periodicities in the energy consumption data, which is typical for high-consumption appliances such as furnace, heater, and A/C. The $ 3 $-way tensor containing the aggregated energy consumption at the considered homes for all months and years is appended to the original tensor along the first mode. Let energy tensor be denoted $ \aggtensor{E} \in \mathbb{R}_+^{I_1+1 \times I_2 \times I_3 \times I_4}$ where $ I_1 $ denotes the number of the devices considered for disaggregation, $ I_2 $ is the number of homes considered, $ I_3 = 12 = $the number of months in a year, and $ I_4 $ denotes the number of years for which data is available.

Notice that we propose the use of the years as a fourth mode instead of dealing with the same months of different years as additional items along the third mode~\cite{batra2018transferring}. Our specific construction of the tensor not only captures seasonal effects, e.g., the energy consumption of HVAC during July and August of a certain year is likely to be similar to other years, but it also results in fewer parameters of the CPD decomposition than the $ 3 $-way construction. This allows for the use for higher ranks in the decomposition while preventing overfitting. Fig.~\ref{fig:agg_tensor} shows a $ 3 $-way tensor which results from considering a single year. The aggregation in this case is a matrix which is appended along the first dimension as depicted.

It is unavoidable that the aggregate energy consumption will be an {\em inexact} aggregation of the consumption of known / monitored devices. This is due to the presence of various (typically low- to medium-power) unmodeled devices/appliances which contribute to the aggregated consumption. Therefore, the energy breakdown application is a case of inexact aggregation, i.e,
$$ \aggtensor{E}(I_1+1, i_2, i_3, i_4) \geq \sum_{i_1=1}^{I_1} \aggtensor{E}(i_1, i_2, i_3, i_4) $$
for all $ i_2 $, $ i_3 $, and $ i_4 $. The next section provides details regarding the simulation setup and presents the algorithm test results.
	\section{Numerical Results}
\label{sec:numerical}
In this section we assess the performance of GRATE in disaggregating aggregated tensor data. 
In the first experiment, we use an exact disaggregation of semi-real energy consumption data. Then, we demonstrate the effectiveness of GRATE in recovering accurate breakdown of aggregated energy consumption data which is aggregated inexactly. The energy disaggregation simulations include two simulation setups which assess the performance of the approach when whole slabs are missing as well as missing fibers, which correspond to unavailability of detailed measurements for a subset of homes, and a subset of months at different homes, respectively.

\subsection{Energy Disaggregation: Evaluation Dataset}

In order to evaluate the performance of the proposed GRATE approach, we utilize the dataset~\cite{parson2015dataport}. Due to the relative novelty of the task, public datasets for multi-year energy consumption are very rare. 
One exception is the Dataport dataset which includes data from homes located in multiple cities in the United States collected over several years. Most of these homes are located in Austin, TX. Therefore, we consider only these homes as it is likely that their energy consumption behavior follows similar patterns. {The data is available in $ 1 $-hour resolution. Since we are interested in disaggregating monthly consumption, we aggregate readings to monthly level in order to evaluate the algorithms.} 

Some of the homes include energy consumption of devices that are not measured at most of the homes. So, we focus only on the measurements related to devices available at most of the homes. The devices considered in our simulations are heating, ventilation, and air conditioning (HVAC), furnace (FN), oven (OV), clothes washer and drier (WD), microwave (MW), and refrigerator (RF). These devices exhibit different energy consumption patterns, i.e., while the HVAC consumption is weather dependent, the WD and MW consumptions are not likely to be dependent on the season or the weather. The number of homes in Austin that include measurements from all these six appliances is $ 103 $, which we choose for our simulations. We sample monthly consumption data from September $ 2014 $ to August $ 2018 $, which includes $ 4 $ years. Therefore, the size of the 4-way tensor, including the aggregate consumption measurements, is $ 7 \times 103 \times 12 \times 4 $.

Note that some homes did not participate in the study for few months during this $ 4 $-year period. Hence, we assume that the data is missing for these months at these homes instead of treating them as zeros. Also, in a few instances, the aggregate monthly consumption appeared to be less than the sum of the measured detailed consumption. In these cases, we consider the aggregate consumption as missing. Fig.~\ref{fig:hist} shows the histogram of the percentage of missing values for the considered users. For example, $ 0.35 $ of the users have less than $ 4\% $ missing readings while almost $ 0.05 $ have more than $ 50\% $ missing values. These unavailable values are treated as missing during all simulations, and are not considered for testing the estimation accuracy as their ground truth values are not available.

\begin{figure}[ht]
	\centering
	\includegraphics[width=0.55\columnwidth]{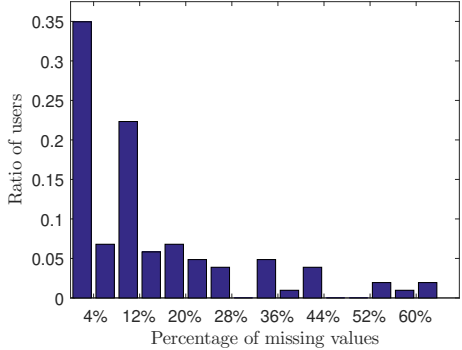}
	\caption{Histogram of percentage of missing values in the original dataset after preprocessing.}
	\label{fig:hist}
\end{figure}

\subsection{Baselines}
Many enegry disaggregation approaches have been proposed in the literature focusing on disaggregating real-time energy consumption. These approaches utilize temporal smoothness of the aggregate readings to identify changes in device energy consumption status in real-time~\cite{almutairi2018scalable}, or otherwise exploit the effect of turning devices on or off to detect changes on the AC waveform measured by smart meters~\cite{saraswat2019}. In this work, we focus on disaggregating monthly consumption data, and hence, we compare against the two most recent approaches dealing with monthly consumption data. 
First, we use the nonnegative matrix factorization approach~\cite{batra2017matrix}. In this method, a matrix is constructed for each appliance which includes the energy consumption of the appliance and the aggregate consumption for all homes. In our case, for each device the size of the matrix is $ 103 \times 96 $ where the rows represent homes, $ 48 $ columns contain the given device's energy consumption for the considered $ 48 $ months, and the remaining $ 48 $ columns represent the aggregated home consumption for all months. This method is called matrix factorization ({MF}) in this section. 

In addition, we use the tensor factorization method proposed in~\cite{batra2018transferring}. The authors used a coupled nonnegative tensor factorization method to analyze data from multiple cities. While the factor related to the appliances is shared between multiple cities, the factors representing months and homes are different from one city to another. Since we only consider data from a single city, Austin, TX, the approach reduces to nonnegative tensor factorization. Different from our proposed method, the tensor in this case is a $ 3 $-way tensor. We refer to this method as nonnegative tensor factorization ({NTF}).

\subsection{Evaluation Metrics}
We adopt the evaluation metric used in~\cite{batra2018transferring}. It measures the quality of distributing the aggregate consumption between the devices compared to the ground-truth breakdown. The metric is called the percentage of energy correctly assigned ({\em PEC}). The definition of PEC for a home, device, month, year ( $ < i_1, i_2, i_3, i_4> $ ) quadruple is given by
\begin{align*}
\text{PEC} ( i_1, i_2, i_3, i_4 ) = \frac{| \aggtensor{E}(i_1, i_2, i_3, i_4) - \hat{\tensor{E}}(i_1, i_2, i_3, i_4) |}{\aggtensor{E}(I_1 + 1, i_2, i_3, i_4)}
\end{align*}
where $ \aggtensor{E}( i_1, i_2, i_3, i_4 ) $ and $ \hat{\tensor{E}}( i_1, i_2, i_3, i_4 ) $ refer to the ground-truth and predicted consumption by device $ i_1 $ in home $ i_2 $ in month $ i_3 $ of year $ i_4 $, respectively. In addition, $ \aggtensor{E}( I_1 + 1, i_2, i_3, i_4 ) $ is the aggregate consumption of home $ i_2 $ in month $ i_3 $ of year $ i_4 $. The appliance root mean squared error in the PEC (ARPEC) for a specific appliance is given by the root mean squared PEC of all missing instances.
In addition, we use the classical normalized mean squared error (NMSE) to measure estimation error of the {held-out} values. Let $ \mathcal{M} $ denote the set of indices of missing values which we aim to estimate. Then, the NMSE is defined as
\begin{align*}
\text{NMSE} = \frac{\sum_{(i_1, \ldots, i_4)^{} \in \mathcal{M}} (  \aggtensor{E}(i_1, \ldots, i_4) - \hat{\tensor{E}}(i_1, \ldots, i_4)  )^2 }{\sum_{(i_1, \ldots, i_4) \in \mathcal{M}} (  \aggtensor{E}(i_1, \ldots, i_4)  )^2}.
\end{align*}

\subsection{Exact Aggregation}
In order to demonstrate the effectiveness of GRATE in the exact aggregation scenario, we construct a tensor where the measured aggregated consumption is replaced with an exact aggregation of the devices' power consumptions. 
To simulate a missing fiber scenario, we randomly sample second mode fibers to hide with different percentages. This means that the energy breakdown is missing for a percentage of all recorded months. 
Our GRATE approach with nonnegativity constraint on the factors is used to disaggregate the monthly consumption and estimate the energy breakdown in each month in the missing fibers. We compare our approach against the nonnegative tensor factorization method which is performed using the Tensorlab toolbox~\cite{tensorlab}. The rank of the tensor used is $ 18 $ which was selected through cross-validation. 

Fig.~\ref{fig:Exact_NMSE} shows the {normalized mean squared error (NMSE)} of estimating held-out values when the percentage of missing fibers is between $ 10\% $ and $ 70\% $. Simulation results show that the GRATE approach achieves better recovery than the NTF especially for high percentages of missing fibers. This shows that utilizing the structure of the compound tensor, as realized by the proposed approach, leads to better estimation performance.

\begin{figure}[t]
	\centering
	\includegraphics[width=0.75\columnwidth]{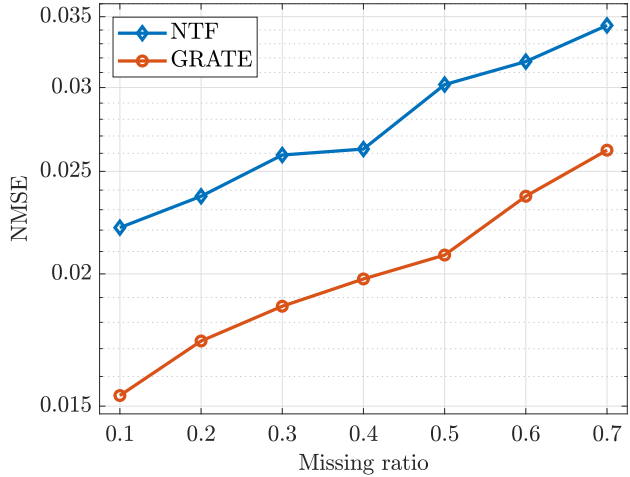}
	\caption{NMSE for different missing fiber percentages}
	\label{fig:Exact_NMSE}
\end{figure}
\subsection{Inexact Aggregation}
Two different scenarios are considered for assessing the performance of the proposed approach in the inexact aggregation case. In the first scenario, we consider the case of missing {detailed data} from homes. In other words, we remove the detailed data of specific homes 
and keep only their aggregated energy consumption. Treating the detailed consumption of their devices as missing, we fit the proposed model as well as the baselines in order to obtain estimates of the energy breakdown. In our simulation, we randomly selected $20$ homes out of total of $103$ homes, and we assume these homes only have aggregated consumption with no device level information. Note that out of all homes, only two homes have device level information available for the whole period of $4$ years.

In the second experiment, we sample first mode fibers from the tensor. For these randomly selected fibers, we treat the detailed consumption of the appliances as missing values. Then, after obtaining the low-rank factors knowing only the aggregated consumption from these fibers, we test the ability of our approach and the baselines to recover the true energy breakdown. In our experiments, for testing, we randomly choose $ 45\% $ of the months with available detailed energy consumption.

The NTF approach is implemented using the Tensorlab toolbox~\cite{tensorlab}. The constrained linear least-squares problems in the iterations of the proposed approach are solved using the Matlab constrained least-squares solver (\texttt{sqlin}). 
For the NTF approach, we select the tensor rank to be $ 18 $ which was validated to be the best rank for both simulation scenarios. For our approach, using a cross-validation technique, we set the rank of the $ 4 $-way tensor to be $ 23 $, which also has almost the same number of parameters as the NTF model with rank $ 18 $. The matrix factorization approach is implemented using Matlab NMF toolbox~\cite{li2013non}. The rank of each matrix factorization is chosen to be {$ 3 $ }in order to have similar number of parameters for all methods. 
Notice that if a rank $ F $ matrix factorization model is used in the MF method for all devices, the total number of parameters in this method is $ 6 F (103 + 96) $.

\subsection{Results and Discussions for Inexact Aggregation}
In the first experiment, we hide the breakdown of energy consumption of $20$ homes for the whole $ 4 $-year period. In Table~\ref{table:NMSE_slabs}, the NMSE of the estimated breakdown for these homes using our approach and the baselines is presented. The proposed GRATE method achieves favorable estimation accuracy for most of appliances. In the few instances where the GRATE does not achieve the best NMSE, it is still very close to the performance of the best baseline. More importantly, the NMSE for estimating the missing slabs of the tensor using the proposed approach is much lower than the baselines.
\begin{table}[h]
	\newcolumntype{C}[1]{>{\centering\arraybackslash}p{#1}}
	\caption[]{NMSE for the energy breakdown for $ 20 $ homes with only aggregated consumption.}
	\begin{center}
		\begin{tabular}{p{1.1cm}  C{1.4cm} C{1.4cm} C{1.4cm}}
			\toprule
			{\em \bf Device}  & {\bf MF} & {\bf NTF} & {\bf GRATE}  \\
			\midrule
			{\em HVAC} &  0.1835 & 0.1258 & {\bf 0.1220}  \\ 
			{\em FN} & 0.3184 & {\bf 0.2233} & {0.2414}\\
			{\em OV} & 0.6957 & {0.6193} & {\bf 0.4705} \\ 
			{\em WD} & 0.3770 &{ 0.2997} & {\bf 0.2463}  \\ 
			{\em MW} & {\bf 0.5919} & {0.6657} & { 0.6462} \\ 
			{\em RF} &  0.2117 & 0.2004 & {\bf 0.1871}  \\ 
			\midrule
			{\em Total} & 0.2937 & 0.3657 & {\bf 0.2271}  \\ 
			\bottomrule
		\end{tabular}
	\end{center}
	\label{table:NMSE_slabs}
\end{table}

In addition, we assess the performance of the proposed approach against the baselines using the RPEC measure. This shows the ratio of the energy that is correctly allocated to a specific appliance. Our approach has an advantage in allocating energy to the correct appliance. Superior performance is observed for the HVAC consumption where the proposed approach has a lower RPEC for both homes.

\begin{table}[h]
	\newcolumntype{C}[1]{>{\centering\arraybackslash}p{#1}}
	\caption[]{ARPEC for the energy breakdown for $ 20 $ homes with only aggregated consumption.}
	\begin{center}
		\begin{tabular}{p{1.1cm}  C{1.4cm} C{1.4cm} C{1.4cm}}
			\toprule
			{\em \bf Device} &  {\bf MF} & {\bf NTF} & {\bf GRATE}  \\
			\midrule
			{\em HVAC} &  0.1335 & 0.0415 & {\bf 0.0352} \\ 
			{\em FN} & 0.1442 & {0.3773} & {\bf 0.0880} \\
			{\em OV} & 0.2745 & {\bf 0.2062} & { 0.2267}  \\ 
			{\em WD} & 0.0649 &{0.0564} & {\bf 0.0475} \\ 
			{\em MW} & 0.0270 & {0.0119} & {\bf 0.0107} \\ 
			{\em RF} &  0.1208 & 0.0946 & {\bf 0.0541}  \\ 
			\bottomrule
		\end{tabular}
	\end{center}
	\label{table:RPEC_slabs}
\end{table}

In the second part of the simulation, we assess the performance of the proposed approach against the baselines in estimating the energy breakdown of $ 45\% $ missing fibers in the tensor. In Table~\ref{table:NMSE_fibers}, the NMSE of the recovered breakdown using all the approaches is presented for each appliance and for the whole missing fibers. The GRATE method achieves superior estimation accuracy especially for the HVAC which consumes most of the aggregated energy. Also, the estimation accuracy of the whole fibers using the GRATE approach is better than the two baselines.

\begin{table}[h]
	\newcolumntype{C}[1]{>{\centering\arraybackslash}p{#1}}
	\caption[]{NMSE for the breakdown of devices consumption in case of $ 45\% $ missing fibers.}
	\begin{center}
		\begin{tabular}{p{1.1cm}  C{1.4cm} C{1.4cm} C{1.4cm}} 
			\toprule
			{\em \bf Device} &  {\bf MF} & {\bf NTF} & {\bf GRATE} \\ 
			\midrule
			{\em HVAC} &  0.1301 & 0.0606 & {\bf 0.0413} \\ 
			{\em FN} & 0.2662 & 0.1551 & {\bf 0.1274}\\
			{\em OV} & 0.5688 & {\bf 0.4761} & 0.4952 \\ 
			{\em WD} &  0.5481 & 0.4088 & {\bf 0.4045} \\ 
			{\em MW} &  0.6748 & {\bf 0.5900} & 0.6049 \\ 
			{\em RF} &  0.2276 & 0.1857 & {\bf 0.1838} \\ 
			\midrule
			{\em Total} & 0.1515 & 0.0777 & {\bf 0.0583} \\ 
			\bottomrule
		\end{tabular}
	\end{center}
	\label{table:NMSE_fibers}
\end{table}

	\section{Conclusions}
\label{sec:conclusion}

In this paper, we presented GRATE, a novel approach for granular recovery of aggregated tensor data using disaggregated examples. The approach formulates the recovery problem as a constrained tensor completion problem where the aggregation structure is enforced through constraints on the latent domain factors. An alternating least-squares solver was employed to tackle the nonconvex optimization problem. The paper focused on the energy disaggregation problem where the consumption of all appliances is aggregated in the monthly bills. The goal is to provide accurate breakdown of energy consumption leveraging the knowledge of disaggregation examples from other homes with smart meters. GRATE was shown to achieve superior recovery performance in different scenarios compared to the state-of-the-art baselines.

	{
	\bibliographystyle{IEEEtran}
	\bibliography{Tensor}
}
\end{document}